\documentclass{article}

\usepackage[english]{babel}

\usepackage[letterpaper,top=2cm,bottom=2cm,left=3cm,right=3cm,marginparwidth=1.75cm]{geometry}

\usepackage{amsmath,amsthm}
\usepackage{graphicx}
\usepackage[colorlinks=true, allcolors = cyan]{hyperref}
\usepackage{todonotes}
\usepackage{cite}

\newtheorem{lemma}{Lemma}
\newtheorem{theorem}{Theorem}
\newtheorem*{remark}{Remark}
\newcommand{\ket}[1]{\left|{#1}\right\rangle}

\newcommand{\ketbra}[2]{{\left|{#1}\right\rangle\!\!\left\langle{#2}\right|}}

\newcommand{\im}{\ensuremath{\texttt{i}}}
\newcommand\comment[1]{}
\newcommand{\new}[1]{{ #1}}
\usepackage{soul}

\title{\Large Representation of the Fermionic Boundary
Operator
}

\author{Ismail Yunus Akhalwaya$^{1}$ 
\and
Yang-Hui He$^2$ \and
Lior Horesh$^3$ \and
Vishnu Jejjala$^4$ \and
William Kirby$^{5}$ \and
Kugendran Naidoo$^4$ \and
Shashanka Ubaru$^3$}

\date{}

\begin{document}
\maketitle

\vspace{-.5cm}
\begin{center}
{\footnotesize
\href{mailto:IsmailA@za.ibm.com}{IsmailA@za.ibm.com},
\href{mailto:hey@maths.ox.ac.uk}{hey@maths.ox.ac.uk},
\href{mailto:lhoresh@us.ibm.com}{lhoresh@us.ibm.com},
\href{mailto:vishnu@neo.phys.wits.ac.za}{vishnu@neo.phys.wits.ac.za},\\
\href{mailto:}{william.kirby@ibm.com},
\href{mailto:9006597F@students.wits.ac.za}{9006597F@students.wits.ac.za},
\href{mailto:Shashanka.Ubaru@ibm.com}{Shashanka.Ubaru@ibm.com}
}
\end{center}

{\small
\noindent $^1$ IBM Research Africa, Johannesburg, 2000, South Africa\\
School of Computer Science and Applied Mathematics, University of the Witwatersrand, Johannesburg, WITS 2050, South Africa

\noindent $^2$ London Institute for Mathematical Sciences, Royal Institution, London W1S 4BS, UK \\
Department of Mathematics, City, University of London, EC1V0HB, UK \\
Merton College, University of Oxford, OX14JD, UK \\
School of Physics, NanKai University, Tianjin, 300071, P.R.\ China

\noindent $^3$ Mathematics of AI, IBM Research, T.J.\ Watson Research Center, Yorktown Heights, NY 10598, USA

\noindent $^4$ Mandelstam Institute for Theoretical Physics, School of Physics, NITheCS, and CoE-MaSS, University of the Witwatersrand, Johannesburg, WITS 2050, South Africa

\noindent $^5$ Department of Physics and Astronomy, Tufts University, Medford, MA 02155, USA\\ IBM Quantum, IBM Research, T.J.\ Watson Research Center, Yorktown Heights, NY 10598, USA
}

${}$

\begin{abstract}

The boundary operator is a linear operator that acts on a collection of high-dimensional binary points (simplices) and maps them to their boundaries. This boundary map is one of the key components in numerous applications, including differential equations, machine learning, computational geometry, machine vision and control systems. We consider the problem of representing the full boundary operator on a quantum computer. We first prove that the boundary operator has a special structure in the form of a complete sum of fermionic creation and annihilation operators. We then use the fact that these operators pairwise anticommute to produce an $O(n)$-depth circuit that exactly implements the boundary operator without any Trotterization or Taylor series approximation errors. Having fewer errors reduces the number of shots required to obtain desired accuracies.

\end{abstract}

\newpage

\section{Introduction}
Quantum computers are capable of performing certain linear algebraic operations in exponentially large spaces, and promise to achieve significant asymptotic speed-ups over classical computers~\cite{feynman1982simulating, hirvensalo2003quantum}. In recent years, several quantum algorithms have been proposed to leverage the potential of quantum computing~\cite{nielsen2010quantum,biamonte2017quantum,schuld2015introduction,schuld2019quantum, lloyd2016quantum, havlivcek2019supervised}. 
These algorithms achieve polynomial to exponential speedups over the best-known classical methods. However, most of these methods require large-scale fault-tolerant quantum computers in order to challenge the classical methods in practice. 

The realization of fault-tolerance in quantum computers is likely at least several years away. Present day quantum computers are referred to as noisy intermediate-scale quantum
(NISQ)~\cite{preskill2018quantum} devices, and are too small to implement error correction, but too large to simulate classically~\cite{zhao2020measurementreduction}. Development of algorithms that achieve quantum advantage for useful tasks on NISQ devices is a critical next step for the field~\cite{aaronson2015read,preskill2018quantum,liu2020rigorous,zhao2020measurementreduction}. 

Boundary operators (also known as boundary maps) are amongst the most important computational primitives used for the representation and analysis of differential equations~\cite{edwards2016differential,berry2014high,lloyd2020quantum}, finite element methods~\cite{zienkiewicz1977finite, montanaro2016quantum}, graph and network analysis~\cite{chung1993laplacian, mo1998study, muhammad2006control}, computational geometry~\cite{brisson1993representing,zomorodian2005computing}, machine vision~\cite{jain1995machine,pietikainen2000texture}, control systems~\cite{fattorini1968boundary, golo2004hamiltonian}, and more. They are also used to bridge the gap between discrete representations, such as graphs and simplicial complexes, and continuous representations, such as vector spaces and manifolds. The graph Laplacians (including higher-order combinatorial Laplacians) can be constructed using the boundary operator ~\cite{hatcher2005algebraic, munkres2018elements}. Laplacians of graphs and hypergraphs play an important role in
spectral clustering~\cite{ng2002spectral}, a computationally tractable solution to the graph partitioning problem that is prevalent in applications such as image segmentation~\cite{shi2000normalized}, collaborative recommendation~\cite{fouss2007random}, text categorization~\cite{kamvar2003spectral}, and manifold learning~\cite{zhang2010clustered}. 
Finally, graph Laplacians are crucial in graph neural network architectures~\cite{kipf2016semi,zhou2020graph,yan2021persistence}, a recent set of neural network techniques for learning graph representations. 

The boundary operator also plays a critical role in Topological Data Analysis (TDA), as a linear operator that acts on a given set of simplices and maps them to their boundaries.
TDA is a powerful machine learning and data analysis technique used to extract shape-related information of large data sets~\cite{zomorodian2005computing,ghrist2008barcodes,bubenik2015statistical,wasserman2018topological}. TDA permits representing large volumes of data using a few global and interpretable features called Betti numbers~\cite{ghrist2008barcodes}. Classical algorithms for TDA and Betti number calculations are typically computationally expensive. In~\cite{lloyd2016quantum}, Lloyd \emph{et al.}~proposed a quantum algorithm for TDA, and this algorithm provably achieves exponential speed up over classical TDA algorithms under certain conditions, as shown in~\cite{gyurik2020towards}.  However, their quantum TDA (or QTDA) algorithm requires fault-tolerance, due to the need for a full Quantum Phase Estimation (made worse by the embedded and repeated use of Grover's search). More recently, Ubaru \emph{et al.}~\cite{ubaru2021quantum} devised a QTDA algorithm that is more amenable to NISQ implementation. The algorithm requires only an $O(n)$-depth quantum circuit, and thus has the potential to be an early useful NISQ algorithm that provably (up to weak assumptions) achieves exponential speed up.

In~\cite{ubaru2021quantum}, the (generalized) boundary operator $B$ is represented in a novel tensor form, as a complete sum of fermionic creation and annihilation operators.\footnote{In~\cite{cade2021complexity}, the connection between supersymmetric many-body systems and the (co)homology problem is discussed, and independently of our work, the same fermionic annihilation operator representation for the (co)boundary map is introduced.} Then, the standard technique of applying a Hermitian operator on a quantum computer is implemented, namely executing the time-evolution unitary $e^{-\im t B}$ for short time $t$ and `solving' for the second term of the Taylor series ($-\im t B$). The time-evolution is approximated by Trotterization~\cite{lloyd1996universal,whitfield2011simulation,childs2018toward}. Using this technique, with certain gate cancellations, an $O(n)$-depth circuit representation can be obtained for the boundary operator~\cite{ubaru2021quantum}.
However, both the Taylor series expansion and Trotterization steps accrue errors and require many shots to accurately simulate $B$.

In this paper, we first present a proof of correctness of the fermionic representation of the boundary operator given in \cite{cade2021complexity,ubaru2021quantum}. Next, we use a technique developed in \cite{izmaylov2020unitarypartitioning,zhao2020measurementreduction} (where it is called \emph{unitary partitioning}) to exactly express the full Hermitian boundary operator as a unitary operator that has an efficient $O(n)$-depth construction in terms of quantum computing primitives.  This short depth circuit, without evolution and Trotterization, analytically implements the  boundary operator, allowing for far fewer measurement shots than the standard Hermitian-evolution technique (see Section \ref{sec:shots}).
This can be instrumental in many downstream applications such as QTDA~\cite{lloyd2016quantum,ubaru2021quantum}, in cohomology problems~\cite{cade2021complexity}, quantum algorithms for finite element methods~\cite{montanaro2016quantum}, solving partial differential equations~\cite{hochbruck2010exponential,berry2014high,lloyd2020quantum}, and potential quantum algorithms for machine vision and control systems.

\section{The Boundary Operator and its Fermionic Representation}

In this section, we first introduce the concept of boundary operators in the context of computational geometry and TDA. We then discuss the fermionic creation and annihilation operators, and present a fermionic representation for the boundary operator along with a proof establishing its correctness.

\subsection{Computational Geometry}
In \cite{lloyd2016quantum}, Lloyd \emph{et al.}~introduce a quantum algorithm for topological data analysis (QTDA). Simplices are represented by strings of $n$ bits. Each bit corresponds to a vertex, with zero indicating exclusion and one inclusion. Hence, on the quantum computer, the usual computational basis directly maps to simplices. There are $2^n$ computational basis vectors, one for each unique $n$-bit binary string, written $\ket{s_k}$, where $k$ indicates the number of vertices in the simplex (\textit{i.e.}, the number of ones in the binary string).\footnote{There is a difference in convention for how to index the bits of the binary string. In this paper, we try to be explicit and refer to `counting from the left' or `right'. For our introduced indices, we always index from 0.} In general, the quantum state vector (of length $2^n$) can be in a superposition of these basis vectors/simplices. Core to the QTDA algorithm is the restricted boundary operator, defined by its action on $k$-dimensional simplices (among $n$ vertices, hence the superscript $^{(n)}$ in the following):
\begin{align}
\label{lloyd_boundary_operator}
  \partial_k^{(n)}  \ket{s_k} &= \sum\limits_l (-1)^l \ket{s_{k-1}(l)} \,,
\end{align}
where $\ket{s_k}$ represents a simplex, and $\ket{s_{k-1}(l)}$ is the simplex of one dimension less than $\ket{s_k}$ with the same vertex set (containing $n$ vertices), but leaving out the $l^{\textrm{th}}$ vertex counting from the left.

\subsubsection{Fleshing out the Restricted Boundary Operator}
\label{s:del}

It helps to rewrite \eqref{lloyd_boundary_operator} without the use of index $l$, which hides some algorithmic steps, because $l$ presupposes that the locations of the ones are known (\textit{i.e.}, which vertices are \emph{in} a given simplex, \textit{e.g.}, $l=0$ refers to the first $1$ in the string reading from left-to-right). We rewrite using index $i$:
\begin{align}
\label{lloyd_boundary_operator_i}
\partial_k^{(n)}  \ket{s_k} &= \sum\limits_{i=0}^{n-1} \delta_{s_k[i], 1}(-1)^{\sum\limits_{j=i+1}^{n-1} s_k[j]} \ket{s_{k-1}(i)} \,,
\end{align}
where we now choose to locate each bit counting from the right starting with index 0. Here we introduce the notation $s_k[i]$ to represent the $i^{\textrm{th}}$ bit of the string. We keep $\ket{s_{k-1}(i)}$ to mean what Lloyd \emph{et al.}\ meant (except now $i$ is any bit index), namely $\ket{s_{k}}$ with the $i^{\textrm{th}}$ vertex set to 0. Crucially all references to $l$ have been replaced, including the implicitly required knowledge of the location of the last $1$ (the sum across $i$ simply runs from $0$ to $n-1$). In Appendix \ref{ap:del}, we illustrate this definition with explicit examples.

Now, $\partial_k^{(n)}$'s action on a single simplex can be split into two cases depending on whether the last vertex, indexed $n-1$, is in the simplex or not. The absence or presence of this last vertex is represented by the bit $s_k[n-1]$ being $0$ or $1$, respectively. If $s_k[n-1]=0$, the quantum state corresponding to this single simplex (call it $\ket{s_{k,0}}$) is a single computational basis state whose binary string has exactly $k$ one's somewhere in it, except that the left most bit is zero. If we write this basis state as an exponentially long vector it consists of a column of $2^{n-1}$ bits only one of which is $1$ (at a location whose binary string consists of those $k$ ones) followed by $2^{n-1}$ zeros. This is similarly the case for $\ket{s_{k,1}}$, but with the single $1$ appearing in the second half of the column vector.
\\
If $\ket{s_k}=\ket{s_{k,0}}$ for $k\leq n-1$,
\begin{align}
\label{lloyd_0}
\partial_k^{(n)}  \ket{s_{k,0}} &= \sum\limits_{i=0}^{n-1} \delta_{s_k[i], 1}(-1)^{\sum\limits_{j=i+1}^{n-1} s_k[j]} \ket{s_{k-1}(i)}\\\nonumber
&= \sum\limits_{i=0}^{n-2} \delta_{s_k[i], 1}(-1)^{\sum\limits_{j=i+1}^{n-2} s_k[j]} \ket{s_{k-1}(i)}\\\nonumber
&= \begin{pmatrix}
\partial_{k}^{(n-1)}&0\\
0&0
\end{pmatrix} \ket{s_{k,0}} \,.
\end{align}
\\
If $\ket{s_k}=\ket{s_{k,1}}$ for $k\leq n$,
\begin{align}
\label{lloyd_1}
\partial_k^{(n)}  \ket{s_{k,1}} &= \sum\limits_{i=0}^{n-1} \delta_{s_k[i], 1}(-1)^{\sum\limits_{j=i+1}^{n-1} s_k[j]} \ket{s_{k-1}(i)}\\ \nonumber
&= \ket{s_{k-1}(n-1)} + (-1)\sum\limits_{i=0}^{n-2} \delta_{s_k[i], 1}(-1)^{\sum\limits_{j=i+1}^{n-2} s_k[j]} \ket{s_{k-1}(i)}\\ \nonumber
&= \begin{pmatrix}
0&P_{k-1}^{(n-1)}\\
0&0
\end{pmatrix} \ket{s_{k,1}} -1\begin{pmatrix}
0&0\\
0&\partial_{k-1}^{(n-1)}
\end{pmatrix} \ket{s_{k,1}} \,,
\end{align}
where $P_{k-1}^{(n-1)}$ is the projection on to the computational basis states of n-1 qubits whose binary strings contain exactly $k-1$ ones. See how the block diagonal notation is allowing us to remove or leave alone the $k^{\textrm{th}}$ one precisely in the $n^{\textrm{th}}$ position. In this particular equation, since we only act on $\ket{s_{k,1}}$ and remove the $n^{\textrm{th}}$ vertex, the projection could be replaced by the identity operator.

\subsection{Fermionic Boundary Operator}
Fermionic fields obey Fermi--Dirac statistics, which means that they admit a mode expansion in terms of creation and annihilation oscillators that anticommute.
Exploiting this fact, it is convenient to map Pauli spin operators to fermionic creation and annihilation operators.
The Jordan--Wigner transformation~\cite{jordan1928paulische} is one such mapping.
In this section, we will make use of it to express the boundary matrix.

The Lloyd \emph{et al.}\ restricted boundary operator given in \eqref{lloyd_boundary_operator} is not in a form that can be easily executed on a quantum computer, nor does it act on all orders, $k$, at the same time. In particular, it is a high-level description of the action of the boundary operator on a single generic $k$-dimensional simplex with the location of the ones assumed to be known. Ubaru \emph{et al.}~\cite{ubaru2021quantum} propose a novel representation that realises the full boundary operator as a matrix. Furthermore, this representation is in tensor product form composed of quantum computing primitives that directly map to quantum gates in the quantum circuit model.
To begin, define the operator:
\begin{equation}
    Q^+ := \frac12 ( \sigma_x +\im\sigma_y) = \left( \begin{array}{cc} 0 & 1 \cr 0 & 0 \end{array} \right) \,.
\end{equation}
This allows Ubaru \emph{et al.}\ to suggest writing the \emph{full} boundary operator in terms of the above operator:
\begin{align}
\label{full_boundary}
\nonumber
        \partial^{(n)} & := \sigma_z \otimes \ldots \otimes \sigma_z \otimes Q^+ \\ \nonumber
        &  + \sigma_z \otimes \ldots \otimes \sigma_z \otimes Q^+ \otimes I \\ \nonumber
        &  \vdots  \\ \nonumber
        &  + \sigma_z \otimes Q^+ \otimes I \otimes \ldots \\ \nonumber
        &  + Q^+ \otimes I \otimes I \otimes \ldots \\
        &  = \sum\limits_{i=0}^{n-1} a_i \,,
\end{align}
where the $a_i$ are the Jordan--Wigner~\cite{jordan1928paulische} Pauli embeddings corresponding to the $n$-spin fermionic annihilation operators, hence the title of this paper.\footnote{While this suggests that only the Jordan--Wigner mapping of fermionic operators does the job, it does open the question of whether other embeddings (\textit{e.g.}, Bravyi--Kitaev\cite{bravyi2002fermionic}) could also play a role, perhaps via translation between the different embeddings. More tantalizingly, since we have recast simplicial homology in terms of these fermionic operators it would be interesting to explore other geometric structures under this light.} To be explicit, $a_i$ is the antisymmetric annihilation operator on mode/orbital $i$ (indexed from the right, starting at zero):

\begin{align}
\label{ai}
a_i :=& \underbrace{\sigma_z \otimes \ldots \otimes \sigma_z}_{n-(i+1)} \otimes\,Q^+ \otimes \underbrace{I \otimes \ldots \otimes I}_{i} \\ \nonumber
=& \sigma_z^{\otimes(n-(i+1))} \otimes Q^+ \otimes I^{\otimes i} \,,\\
\{a_i,a_j^\dagger\}=& \delta_{ij}I^{\otimes n}~.
\end{align}
\noindent
When \eqref{full_boundary} is realised in matrix form, where the rows' and columns' integer indices (counting from left to right and top to bottom, beginning at zero) are translated into binary, the meaning of these strings correspond exactly to the simplex strings introduced above. Recall, a one at bit $i$ of the simplex string (counting from the right, starting from zero) corresponds to selecting the $i^{\textrm{th}}$ vertex.

The first key result we need to show is:
\begin{theorem}
The full boundary operator as defined in \eqref{full_boundary} is a sum over the restricted boundary operators over simplices of each dimension as defined in \S\ref{s:del} as follows
\begin{equation}
\label{rts}
\partial^{(n)} = \sum\limits_{k=1}^{n} \partial_k^{(n)}.
\end{equation}
\end{theorem}

\begin{remark}
A simple sum (as opposed to a direct sum) can be used because the Lloyd \emph{et al.}\ definition of $\partial_k$ has already been extended to a common vertex space of size $n$. The rewriting constructs the full boundary operator in one step and the individual $\partial_k$'s do not feature at all. Indeed to recover the $\partial_k$'s one would need to project into the appropriate space, which is a required step in implementing QTDA on quantum computers\cite{ubaru2021quantum}. The index $i$ of the $a_i$ fermionic operators indicates which qubit out of $n$ is acted upon by $Q^+$ (counting from the right) and has nothing to do with $\partial_k$ (and is why there is no $k$ dependence).
\end{remark}

\begin{proof}
The proof proceeds by induction, so we will first set up a recurrence relation.
\paragraph{Recurrence Relation: }
Equation \eqref{full_boundary} is amenable to efficient quantum circuit construction as we show in Section \ref{unitary_circuit}. However, first we need to prove that equation \eqref{full_boundary} is valid and correctly implements Lloyd \emph{et al.}'s high level description \eqref{rts}. The easiest way to prove equation \eqref{rts} is by noticing that the left hand side satisfies the following recurrence relation and then connecting the recurrence relation to the right hand side of \eqref{rts}:
\begin{align}
\label{recurrence}
    \partial^{(n)} &= Q^+ \otimes I^{\otimes(n-1)} + \sigma_z \otimes \partial^{(n-1)} \,,\\\nonumber
    \partial^{(1)} &= Q^+ \,.
\end{align}
It helps to write \eqref{recurrence} in block diagonal form:
\begin{equation}
\label{block_recurrence}
    \partial^{(n)} = 
    \begin{pmatrix}
   \partial^{(n-1)}  & I^{\otimes(n-1)} \\
    0 & -\partial^{(n-1)} \\
    \end{pmatrix} 
\end{equation}
and think about the action of $\partial^{(n)}$ in terms of $\partial^{(n-1)}$.

The operator $\partial^{(n-1)}$ acts on vectors of size $2^{n-1}$, while $\partial^{(n)}$ acts on vectors of size $2^{n}$. The vector of size $2^n$ can be seen as two halves of size $2^{n-1}$. The top half consists of simplices where the $n^{\textrm{th}}$ vertex is not in the simplex. The bottom half corresponds to a `copy' of the $n-1$ space but now with the $n^{\textrm{th}}$ vertex selected. For example, edges in the upper half become triangles in the lower half, since the $n^{\textrm{th}}$ vertex is added. The top-left block of $\partial^{(n)}$ in \eqref{block_recurrence} is $\partial^{(n-1)}$ acting on the `top-half' simplices, where the $n^\textrm{th}$ vertex is not selected ($s_k[n-1]=0$ and $k$ has to be strictly less than $n$) and returns simplices (the boundaries) in the top half (since taking the boundary can never add the $n^\textrm{th}$ vertex). The top-right block ($I^{\otimes(n-1)}$) acts on simplices in the bottom half ($s_k[n-1]=1$) and returns simplices in the top half, which corresponds to simply removing the $n^\textrm{th}$ vertex and leaving everything else the same (hence the identity operator). Finally, the bottom-right block acts on simplices in the bottom half ($s_k[n-1]=1$) and returns simplices in the bottom half, corresponding to those returned by taking the boundary operator acting on the first $n-1$ vertices and \emph{leaving} the $n^\textrm{th}$ vertex selected. The bottom returned vertices need a multiplication by $-1$ because acting on the first $n-1$ vertices by $\partial^{(n-1)}$ does not factor in the extra minus one from the $n^\textrm{th}$ vertex in the definition \eqref{lloyd_boundary_operator_i}.

\paragraph{Inductive Proof of Simplex Action: }
We now prove \eqref{rts} via induction.
The base case, $\partial^{(1)} = \partial_1^{(1)}$, clearly implements the Lloyd \emph{et al.}\ definition because it takes the single vertex vector to the null vector as required. Now we assume $\partial^{(n)}$ for $n\geq1$ implements the sum of restricted boundary operators \eqref{rts} up to $n$, and use that to show that $\partial_{n+1}$ correctly implements the sum up to $n+1$, i.e., that
\begin{equation}\label{rts_induction}
\partial^{(n+1)} = \sum\limits_{k=1}^{n+1} \partial_k^{(n+1)}\,.
\end{equation}
From \eqref{block_recurrence}, the left hand side of \eqref{rts_induction} is
\begin{equation}\partial^{(n+1)} =     \begin{pmatrix}
    \partial^{(n)}  & I^{\otimes(n)} \\
    0 & -\partial^{(n)} \\
    \end{pmatrix} \,, 
\end{equation}
and from the inductive assumption \eqref{rts}, this becomes
\begin{equation}
\label{lhs_final}
    \partial^{(n+1)} =     \begin{pmatrix}
    \sum\limits_{k=1}^{n} \partial_k^{(n)} & I^{\otimes n} \\
    0 & -\sum\limits_{k=1}^{n} \partial_k^{(n)} \\
    \end{pmatrix} \,.
\end{equation}

The right hand side of \eqref{rts_induction} contains terms of the form $\partial_k^{(n+1)}$. All $\partial_k^{(n+1)}$ act on $n+1$ vertices, so in matrix form they are all of the same dimension.
From \eqref{lloyd_0} and \eqref{lloyd_1} we have
\begin{align}
\label{lloyd_0_second}
\partial_k^{(n+1)} \ket{s_{k,0}}
&= \begin{pmatrix}
\partial_{k}^{(n)}&0\\
0&0
\end{pmatrix} \ket{s_{k,0}}, \quad 1\leq k\leq n \,, \\
\label{lloyd_1_second}
\partial_k^{(n+1)}  \ket{s_{k,1}}
&= \begin{pmatrix}
0&P_{k-1}^{(n)}\\
0&-\partial_{k-1}^{(n)}
\end{pmatrix} \ket{s_{k,1}}, \quad 1\leq k\leq n+1 \,.
\end{align}
Now since any $\ket{s}=\sum_k\ket{s_k}= \sum_k (\ket{s_{k,0}} + \ket{s_{k,1}})$ and the matrix in \eqref{lloyd_0_second} takes $\ket{s_{k,1}}$ to the null vector: $\begin{pmatrix}
\partial_{k}^{(n)}&0\\
0&0
\end{pmatrix} \ket{s_{k,1}} = 0$, while the matrix in \eqref{lloyd_1_second} similarly takes $\ket{s_{k,0}}$ to the null vector: $\begin{pmatrix}
0&P_{k-1}^{(n)}\\
0&-1 \cdot \partial_{k-1}^{(n)}
\end{pmatrix} \ket{s_{k,0}} = 0$,
we can combine \eqref{lloyd_0_second} and \eqref{lloyd_1_second}, while paying attention to different values of $k$:
\begin{align}
\partial_1^{(n+1)} &= \begin{pmatrix}
\partial_{1}^{(n)}&P_{0}^{(n)}\\
0&0
\end{pmatrix}, \quad k=1 \,, \\
\partial_k^{(n+1)}  &= \begin{pmatrix}
\partial_{k}^{(n)}&P_{k-1}^{(n)}\\
0&-\partial_{k-1}^{(n)}
\end{pmatrix}, \quad 2 \leq k \leq n \,, \\
\partial_{n+1}^{(n+1)}  &= \begin{pmatrix}
0&P_{n}^{(n)}\\
0&- \partial_{n}^{(n)}
\end{pmatrix}, \quad k = n+1 \,,
\end{align}
where $P_{0}^{(n)}=\ketbra{0\cdots 0}{0\cdots 0}$ ($n$ zeroes) and $P_{n}^{(n)}=\ketbra{1\cdots 1}{1\cdots 1}$ ($n$ ones). Note that $\partial_{n+1}^{(n+1)}$ produces no top-left block and $\partial_{1}^{(n+1)}$ produces no bottom-right block.

Therefore, the right hand side of \eqref{rts_induction} is:
\begin{align}
    \sum\limits_{k=1}^{n+1} \partial_k^{(n+1)} &=      \begin{pmatrix}
    \sum\limits_{k=1}^{n} \partial_k^{(n)} & I^{\otimes n} \\
    0 & -\sum\limits_{k=1}^{n} \partial_k^{(n)} \\
    \end{pmatrix},
\end{align}
since $\sum_{k=0}^{n} P_{k}^{(n)} = I^{\otimes n}$. Therefore the right-hand side is equal to the left-hand side \eqref{lhs_final}.
\end{proof}

With the proof in hand, we now understand why \eqref{full_boundary} and its recursive form \eqref{recurrence} can only produce \eqref{lloyd_boundary_operator} in summation form \eqref{rts}. In particular, $\partial_k$ on a simplex with one more potential vertex ($n+1$) is related to both $\partial_k$ and $\partial_{k-1}$ on the original number of vertices ($n$). Therefore, the full chain of sums is needed to build up a recursive construction, where we sequentially add one vertex at a time, which then makes a quantum-implementable tensor definition possible.

\section{Unitary circuit}
\label{unitary_circuit}
\new{The boundary operator above is not unitary. It is also not Hermitian which makes it more difficult to use known techniques to implement it on a quantum computer. As a first step we consider a Hermitian version of the boundary operator constructed by simply adding the Hermitian conjugate. It turns out that this Hermitian version happens to be a scaled unitary operator, therefore requiring no further work to obtain a unitary circuit. Naturally, the Hermitian version is not identical to the non-Hermitian version but a simple projection-based reconstruction is discussed below.}

Every $Q^+$ in the tensor product now becomes $\sigma_x = Q^+ + (Q^+)^\dagger$.
Hence, the full Hermitian boundary operator is
\begin{equation}\label{full_B}
    \begin{split}
        B = \partial + \partial^\dagger & = \sigma_z \otimes \ldots \otimes \sigma_z \otimes \sigma_x \\
        &  + \sigma_z \otimes \ldots \otimes \sigma_z \otimes \sigma_x \otimes I \\
        &  \vdots  \\
        &  + \sigma_z \otimes \sigma_x \otimes I \otimes \ldots \\
        &  + \sigma_x \otimes I \otimes I \otimes \ldots \\
& =
\sum\limits_{i=0}^{n-1}Q_i \,,
    \end{split}
\end{equation}
where $I$ denotes the single-qubit identity operator, and
\begin{equation}
    Q_i := a_i + a_i^\dagger \, ,
\end{equation}
with $a_i$ as defined in \eqref{ai}.
This, $Q_i$ is a Kronecker product of $n$ Pauli matrices, of which the $n-(i+1)$ leftmost are $\sigma_z$.
The goal is now to construct a unitary circuit to implement either $B$ or $\exp(\im B t)$ (depending on the application). Towards this end, it is extremely useful to note that:

\begin{lemma}
The $Q_i$'s pairwise anticommute:
$\{Q_i, Q_j\} = 0$ for $i \neq j$.
\end{lemma}
\begin{proof}
The mixed-product identity for the Kronecker product is
\begin{equation}
    (A \otimes B) \cdot (C \otimes D) = (AC)\otimes (BD) \,.
\end{equation}
Hence, we have (assuming without loss of generality that $i > j$ so that there are more $\sigma_z$ matrices in $Q_j$)
\begin{equation}
\label{pauli_anticomm}
\{Q_i, Q_j\} = \big\{\underbrace{\sigma_z \otimes \cdots \otimes \sigma_z}_{n-(i+1)} \otimes\,\sigma_x \otimes \underbrace{I\otimes \cdots \otimes I}_i\ , \quad \underbrace{\sigma_z \otimes \cdots \otimes \sigma_z}_{n-(j+1)}  \otimes\,\sigma_x  \otimes \underbrace{I\otimes \cdots \otimes I}_j \,\big\}.
\end{equation}
Since $\{\sigma_\ell, \sigma_k \} = 2 \delta_{\ell k} \mathbf{I}$ and $[\sigma_k, I] = 0$ for any $k=x,y,z$, \eqref{pauli_anticomm} immediately yields the result.
\end{proof}
We now use the fact that any real linear combination of pairwise anticommuting Pauli operators is unitarily equivalent to a single Pauli operator, up to some rescaling~\cite{izmaylov2020unitarypartitioning}.
One can think of this as a generalization of the Bloch sphere to more than three Pauli operators.
Moreover, the unitary that maps the linear combination to the single Pauli operator can be efficiently constructed.
This technique was developed for the purpose of reducing the number of distinct terms in Hamiltonians to be simulated using variational quantum eigensolvers, in which context it is known as \emph{unitary partitioning}~\cite{izmaylov2020unitarypartitioning,zhao2020measurementreduction,ralli2021measurementreduction}.

For neighbouring $Q$'s $Q_{i-1}$ and $Q_i$,
\begin{equation}
\begin{split}
    -\im Q_{i-1}Q_i & =
    -\im(\underbrace{\sigma_z \otimes\cdots\otimes \sigma_z}_{n-i}\otimes\,\sigma_x \otimes \underbrace{I \otimes \ldots \otimes I}_{i-1})\cdot (\underbrace{\sigma_z \otimes\cdots\otimes \sigma_z}_{n-(i+1)} \otimes\,\sigma_x \otimes \underbrace{I \otimes \ldots \otimes I}_i)\\
 & = -\im I^{\otimes(n-(i+1))} \otimes \sigma_z \sigma_x \otimes \sigma_x \otimes I^{\otimes(i-1)}
    \\
    & = -\im I^{\otimes(n-(i+1))} \otimes\im\sigma_y \otimes \sigma_x \otimes I^{\otimes(i-1)}
    \\
    & = Y_{i}X_{i-1} \,,
    \end{split}
\end{equation}
where $X_{i-1}$ and $Y_{i}$ are, respectively,
\begin{eqnarray}
X_{i-1} &=& I^{\otimes(n-i)}\otimes\,\sigma_x \otimes I^{\otimes(i-1)}, \\
Y_{i} &=& I^{\otimes(n-(i+1))} \otimes\,\sigma_y \otimes I^{\otimes i}.
\end{eqnarray}

Hence, $-\im Q_{i-1}Q_i$ is itself a Pauli operator, and commutes with all Pauli terms in $B$ except for $Q_{i-1}$ and $Q_{i}$.
Therefore, a rotation generated by $-\im Q_{i-1}Q_i$ only affects those two terms.
For an arbitrary linear combination $\alpha Q_{i}+\beta Q_{i-1}$ for real $\alpha$ and $\beta$, define the following rotation generated by $-\im Q_{i-1}Q_i$:
\begin{equation}
\label{adjoint_action}
    R_{i}\equiv\exp\left(\frac{Q_{i-1}Q_i}{2}\text{atan2}(\alpha,\beta)\right)=\exp\left(\frac{\im Y_{i}X_{i-1}}{2}\text{atan2}(\alpha,\beta)\right) \, .
\end{equation}
It can be shown that the adjoint action of $R_i$'s on the linear combination is
\begin{equation}
    R_{i}(\alpha Q_{i}+\beta Q_{i-1})R_{i}^\dagger=\sqrt{\alpha^2+\beta^2}Q_{i-1} \,.
\end{equation}
For details, we refer the reader to the unitary partitioning papers~\cite{izmaylov2020unitarypartitioning,zhao2020measurementreduction,ralli2021measurementreduction}.

Therefore, we can map $B$, which as given in \eqref{full_B} is the sum of all of the $Q_i$s, to a single $Q_i$ via a sequence of such rotations as follows:
\begin{equation}
\begin{split}
    B=\sum\limits_{i=0}^{n-1}Q_i\quad&\xrightarrow{R_{n-1}\text{ with $\alpha=1$ \,, $\beta=1$}}\quad\sqrt{2}Q_{n-2}+\sum\limits_{i=0}^{n-3}Q_i\\
    &\xrightarrow{R_{n-2}\text{ with $\alpha=\sqrt{2}$, $\beta=1$}}\quad\sqrt{3}Q_{n-3}+\sum\limits_{i=0}^{n-4}Q_i\\
    &\xrightarrow{R_{n-3}\text{ with $\alpha=\sqrt{3}$, $\beta=1$}}\quad\sqrt{4}Q_{n-4}+\sum\limits_{i=0}^{n-5}Q_i\\
    &\qquad\vdots\\
    &\xrightarrow{R_{2}\text{ with $\alpha=\sqrt{n-2}$, $\beta=1$}}\quad\sqrt{n-1}Q_1+Q_0\\
    &\xrightarrow{R_{1}\text{ with $\alpha=\sqrt{n-1}$, $\beta=1$}}\quad\sqrt{n}Q_0 \,,
\end{split}
\end{equation}
where each arrow represents an application of $R_i$ as in \eqref{adjoint_action} with the given values of $\alpha$ and $\beta$.
Hence, we map $B$ to $\sqrt{n}Q_0=\sqrt{n}X_0$, \textit{i.e.}, a single-qubit Pauli $\sigma_x$, via $n-1$ rotations generated by two-qubit Paulis $Y_{i}X_{i-1}$.

Let $R$ denote this entire sequence of rotations, \textit{i.e.},
\begin{equation}
\label{RRi}
R = \prod_{i=(n-1)}^{i=1} R_i \,.
\end{equation}
In terms of $R$, the above result is
\begin{equation}
    RBR^\dagger=\sqrt{n}X_0 \,,
\end{equation}
which implies that
\begin{equation}
\label{finalboundary}
    B=\sqrt{n} R^\dagger X_0 R \,.
\end{equation}
This is all that is needed in the QTDA use case. Since $R^\dagger X_0 R$ is unitary, it can be implemented as a quantum circuit, and to obtain $B$ the constant of proportionality can be included during classical postprocessing.

\paragraph{Extensions:}
\new{ In certain} use cases (such as propagators in differential equations), it may be desitred to implement the time evolution of $B$. For these cases, in order to implement the time-evolution generated by $B$ for a time $t$, we can apply $R$, then implement the time-evolution generated by $\sqrt{n}X_0$ for time $t$ (analytically), and then invert $R$. \new{This exponentiation is achieved analytically and therefore does not suffer from any Trotterization error.}:
\begin{equation}
\label{expB}
    e^{-\im Bt}=R^{\dagger} e^{-\im\sqrt{n}X_0t}R \,,
\end{equation}
where $R$ and $R_i$ are as defined in \eqref{adjoint_action} and \eqref{RRi}. The cost of implementing this evolution is independent of $t$, since we can classically precompute $\sqrt{n}t\mod2\pi$ and implement the rotation generated by $X_0$ through this angle.

\new{Finally, in applications where we require the non-Hermitian boundary operator $\partial_k$  (e.g., computational geometry), we can apply projection operators on either side of $B$ to compute $\partial_k = P_{k-1}BP_k$, see~\cite{ubaru2021quantum} for details on how to construct these projectors efficiently on a quantum computer.}

\section{Discussion: Circuit Depth and Shots}
\label{sec:shots}
The above unitary form of the boundary operator \eqref{finalboundary} has depth $O(n)$ since there are $2 (n-1)$ rotations and one $X_0$. To be precise, the dependence on $n$ is contained in $R$, which is a sequence of $n-1$ two-qubit rotations, so in total we require $2(n-1)$ two-qubit rotations.
To finalize the $O(n)$ discussion, we must only explain that each Pauli rotation (involving two qubits) can be implemented in constant depth. For example, the standard way to evolve a single Pauli string is to change the basis of each of the qubits affected by a $\sigma_x$ or $\sigma_y$ into the $\sigma_z$ basis (we only have two such qubits for each rotation, which is independent of $n$) followed by rotation around the $z$-axis while accounting for parity. Therefore the two-qubit rotations are independent of $n$ and the overall depth of \eqref{finalboundary} is $O(n)$.

\new{With the boundary operator performed analytically there is a significant saving in the number of shots needed compared to the approximate method of the original fermionic boundary operator \eqref{full_B} as used in  $\Delta$ \cite{ubaru2021quantum}. 
In order to account for the Trotterization error, which for an $n$-qubit Hamiltonian written as $B = \sum_{i=0}^{n-1} Q_i$, we have the first order approximation 
$
e^{-i\sum_{i=0}^{n-1} Q_i t} =\prod_{i=0}^{n-1} e^{-iQ_i t} + O( t^2)$. Therefore the error due Trotterization is $\epsilon_T \sim O(t^2)$. 

Then the Taylor series expansion for the Hamiltonian simulation is given by $ e^{-iBt} = 1 - iBt - B^2t^2/2 + O(B^3t^3) = 1 - iBt  + O(nt^2)$, since $B^2 = nI$. Solving for $B$ yields: $ B = (e^{-iBt}+ i+ \epsilon_{TS})/t$, where $\epsilon_{TS} \sim O(nt^2)$.
Therefore, including $\epsilon_T$, 
\[B = \left(\prod_{i=0}^{n-1} e^{-iQ_i t}+ i + \epsilon_{TS} + \epsilon_T\right)/t.\]

Next, taking the expectation produces shot noise, $\epsilon_{\textrm{shot}}\sim\frac{1}{\sqrt{N}}$. Similar to $\epsilon_T$ and $\epsilon_{TS}$, $\epsilon_{\textrm{shot}}$ gets amplified by the fractional $t$ denominator: $\epsilon_{\textrm{shot}}/t$. Therefore, the overall error is:
 \[\epsilon = (\epsilon_T + \epsilon_{TS} + \epsilon_{\textrm{shot}})/t.\]
Now, each of these errors should be of the same order as the desired order of the precision $\epsilon$. Handling each error independently and solving for $t$ in the stronger Taylor series constraint, we have $t \sim O(\epsilon/n)$. From the shot noise constraint and solving for $N$ we have: $N \sim O(1/(\epsilon^2t^2))$. Finally, inserting the above $t$ dependence on $\epsilon$ yields $N \sim O(n^2/\epsilon^4)$ for the approximate version of the original fermionic boundary operator\cite{ubaru2021quantum}.


In contrast, for the analytic circuit presented in this paper, since there is no Trotterization and Taylor expansion error, the number of samples needed is $N \sim  O(\frac{1}{\epsilon^{2}})$, which is a quadratic saving (and more for higher moments).}

\section{Conclusion}
In this paper we provided a short-depth, $O(n)$, quantum circuit for the exact, analytical implementation of the full Hermitian boundary operator on a gate-based quantum computer. This is a significant improvement over previous approximate implementations, resulting in at least quadratic savings in the number of shots. In order to achieve this we formally proved that the boundary operator can be written as a sum of fermionic creation and annihilation operators. This connection between algebraic geometry and fermionic operators opens a potentially rich vein for further exploration. The fermionic representation together with the convenient property of pairwise Pauli anticommutation permits implementation via a short circuit consisting of a cascade of two-qubit rotations. With such a short-depth circuit for the full boundary operator, the door is open for the search for quantum implementations of many algorithms that have the boundary operator as a core component, including in the fields of differential equations and machine learning.

\section*{Acknowledgements}
YHH would like to thank UK STFC for grant ST/J00037X/2.
VJ is supported by the South African Research Chairs Initiative of the Department of Science and Technology and the National Research Foundation and by the Simons Foundation Mathematics and Physical Sciences Targeted Grant, 509116.
WMK is supported by the National Science Foundation, Grant No. DGE-1842474.

\appendix
\section{Explicit Examples}\label{ap:del}
In this appendix, we illustrate the formulae in the main text with the explicit example of $n=3$.
There are three vertices here, which we will call $v_{0,1,2}$.
We then use binary degree-lexicographic representation of the 8 simplices constructible from these vertices as follows:
\begin{equation}
\mbox{
\begin{tabular}{c|c|c}
Dimension & Binary Representation & Simplex \\ \hline
- & 000 & $\emptyset$  \\ \hline
0 & 001 & $v_0$ \\
  & 010 & $v_1$ \\
  & 100 & $v_2$ \\ \hline
1 & 011 & line$(v_0,v_1)$	\\
  & 101 & line$(v_0,v_2)$	  \\
  & 110 & line$(v_1,v_2)$	  \\ \hline
2 & 111 & triangle$(v_0,v_1,v_2)$ \\
\end{tabular}
}
\end{equation}

Let us check a few of the boundary operators as defined in \S\ref{s:del}.
Consider $\partial_1^{(3)}$, which should take a line to its boundary endpoints, with sign.
Take $\ket{s_{1}} = 011$, we have only 2 terms contribution since the Kronecker-delta picks out $v_1$ and $v_0$, \textit{i.e.}, $i=0,1$
\new{\begin{eqnarray}
    \partial_1^{(3)} \ket{s_{1}} &=&
    (-1)^{\sum\limits_{j=1}^{2} s_1(j)} v_0
    +
    (-1)^{\sum\limits_{j=2}^{2} s_1(j)} v_1
    \\\nonumber
    &= &
    (-1)^{1+0} v_0 + (-1)^{0} v_1
    \\\nonumber
    &= &v_1 - v_0 \ .
\end{eqnarray}}
Next, take $\ket{s_{2}} = 111$ and act upon it with $\partial_2^{(3)}$, whereupon all 3 terms in the sum contribute:
\new{\begin{eqnarray}
    \partial_2^{(3)} \ket{s_{2}} &=&
    (-1)^{\sum\limits_{j=1}^{2} s_2(j)}
    v_0
    +
    (-1)^{\sum\limits_{j=2}^{2} s_2(j)}
    v_1
    +
    (-1)^{\sum\limits_{j=3}^{2} s_2(j)}
    v_2\\\nonumber
    &= &
    (-1)^{1+1}v_0 + (-1)^{1}v_1 + (-1)^0 v_2
    \\\nonumber
    &= & v_0 - v_1 + v_2 \ ,
\end{eqnarray}}
which says that that the boundary of the triangle are the 3 (signed) edges.

\subsection{The Boundary Operator $\partial^{(n)}$}
We now write down the boundary operator $\partial^{(n)}$, constructed from the pieces $\partial^{(n)}_k$ given in \eqref{full_boundary}, explicitly.
The initial cases are simply:
\begin{align}
\nonumber
    \partial^{(1)} 
    &= \left( \begin{array}{cc} 0 & 1 \cr 0 & 0 \end{array} \right) ~,
    \\
\nonumber
    \partial^{(2)} &= a_0 + a_1
    = \sigma_z \otimes Q^{+} + Q^{+} \otimes I =
    {\scriptsize
\left(
\begin{array}{cccc}
 0 & 1 & 0 & 0 \\
 0 & 0 & 0 & 0 \\
 0 & 0 & 0 & -1 \\
 0 & 0 & 0 & 0 \\
\end{array}
\right) +
\left(
\begin{array}{cccc}
 0 & 0 & 1 & 0 \\
 0 & 0 & 0 & 1 \\
 0 & 0 & 0 & 0 \\
 0 & 0 & 0 & 0 \\
\end{array}
\right)
}
 = 
 {\scriptsize
\left(
\begin{array}{cccc}
 0 & 1 & 1 & 0 \\
 0 & 0 & 0 & 1 \\
 0 & 0 & 0 & -1 \\
 0 & 0 & 0 & 0 \\
\end{array}
\right) ~,
}
\\
\nonumber
\partial^{(3)} &= 
a_0 + a_1 + a_2 
= \sigma_z \otimes \sigma_z \otimes Q^{+} 
+ \sigma_z \otimes Q^{+}  \otimes I
+ Q^{+} \otimes I \otimes I
\\
\nonumber
& =
{\tiny
 \left(
\begin{array}{cccccccc}
 0 & 1 & 0 & 0 & 0 & 0 & 0 & 0 \\
 0 & 0 & 0 & 0 & 0 & 0 & 0 & 0 \\
 0 & 0 & 0 & -1 & 0 & 0 & 0 & 0 \\
 0 & 0 & 0 & 0 & 0 & 0 & 0 & 0 \\
 0 & 0 & 0 & 0 & 0 & -1 & 0 & 0 \\
 0 & 0 & 0 & 0 & 0 & 0 & 0 & 0 \\
 0 & 0 & 0 & 0 & 0 & 0 & 0 & 1 \\
 0 & 0 & 0 & 0 & 0 & 0 & 0 & 0 \\
\end{array}
\right)
 }
+
{\tiny
\left(
\begin{array}{cccccccc}
 0 & 0 & 1 & 0 & 0 & 0 & 0 & 0 \\
 0 & 0 & 0 & 1 & 0 & 0 & 0 & 0 \\
 0 & 0 & 0 & 0 & 0 & 0 & 0 & 0 \\
 0 & 0 & 0 & 0 & 0 & 0 & 0 & 0 \\
 0 & 0 & 0 & 0 & 0 & 0 & -1 & 0 \\
 0 & 0 & 0 & 0 & 0 & 0 & 0 & -1 \\
 0 & 0 & 0 & 0 & 0 & 0 & 0 & 0 \\
 0 & 0 & 0 & 0 & 0 & 0 & 0 & 0 \\
\end{array}
\right)
 }
+
{\tiny
\left(
\begin{array}{cccccccc}
 0 & 0 & 0 & 0 & 1 & 0 & 0 & 0 \\
 0 & 0 & 0 & 0 & 0 & 1 & 0 & 0 \\
 0 & 0 & 0 & 0 & 0 & 0 & 1 & 0 \\
 0 & 0 & 0 & 0 & 0 & 0 & 0 & 1 \\
 0 & 0 & 0 & 0 & 0 & 0 & 0 & 0 \\
 0 & 0 & 0 & 0 & 0 & 0 & 0 & 0 \\
 0 & 0 & 0 & 0 & 0 & 0 & 0 & 0 \\
 0 & 0 & 0 & 0 & 0 & 0 & 0 & 0 \\
\end{array}
\right)
} \\
&=
{\tiny
\left(
\begin{array}{cccccccc}
 0 & 1 & 1 & 0 & 1 & 0 & 0 & 0 \\
 0 & 0 & 0 & 1 & 0 & 1 & 0 & 0 \\
 0 & 0 & 0 & -1 & 0 & 0 & 1 & 0 \\
 0 & 0 & 0 & 0 & 0 & 0 & 0 & 1 \\
 0 & 0 & 0 & 0 & 0 & -1 & -1 & 0 \\
 0 & 0 & 0 & 0 & 0 & 0 & 0 & -1 \\
 0 & 0 & 0 & 0 & 0 & 0 & 0 & 1 \\
 0 & 0 & 0 & 0 & 0 & 0 & 0 & 0 \\
\end{array}
\right) ~.
}
\end{align}

\subsubsection{Nilpotency}
Let us check that the boundary operator is nilpotent, as is required from homology:
\begin{lemma}
The boundary operator satisfies
$\partial^{(n)} \cdot \partial^{(n)} = 0$.
\end{lemma}
\begin{proof}
we can proceed by induction.
Immediately, one checks that $(\partial^{(1)})^2 = (\partial^{(2)})^2 = (\partial^{(3)})^2 = 0$ where $0$ is the $2^n \times 2^n$ matrix of zeros; thus, the initial terms are fine.
Next, assume that $(\partial^{(n)})^2 = 0$ as the induction hypothesis.
Now, we have (the third line uses the so-called mixed product rule for Kronecker and dot products): 
\begin{align}
\nonumber
(\partial^{(n+1)})^2 &= \left( \partial^{(n)} \otimes I + \sigma_z^{\otimes n} \otimes Q^{+}  \right)^2 \\
\nonumber
	 &= \left( \partial^{(n)} \otimes I \right)^2 + \left(\sigma_z^{\otimes n} \otimes Q^{+}  \right)^2
+  (\partial^{(n)} \otimes I) \cdot (\sigma_z^{\otimes n} \otimes Q^{+} ) + 
(\sigma_z^{\otimes n} \otimes Q^{+} ) \cdot (\partial^{(n)} \otimes I) \\
\nonumber
	&= (\partial^{(n)})^2 \otimes I^2 + \left(\sigma_z^{\otimes n}\right)^2 \otimes (Q^{+})^2 
	+ (\partial^{(n)} \cdot \sigma_z^{\otimes n}) \otimes (I \cdot Q^{+} ) + \\
\nonumber
&  \quad (\sigma_z^{\otimes n} \cdot \partial^{(n)} )  \otimes  ( Q^{+}  \cdot I ) 
\qquad
\mbox{[since $(A \otimes B) \cdot (C \otimes D) = (AC)\otimes (BD)$, the mixed-product identity]} \\
& = 0 + 0 + (\sigma_z^{\otimes n} \cdot \partial^{(n)} + \partial^{(n)} \cdot \sigma_z^{\otimes n}) \otimes v
\qquad
\mbox{[since $(\partial^{(n)})^2=(Q^{+})^2=0$]} ~.
\end{align}
Thus it suffices to prove the the anticommutation
\begin{equation}
\left\{ \sigma_z^{\otimes n}, \  \partial^{(n)} \right\} = \sigma_z^{\otimes n} \cdot \partial^{(n)} + \partial^{(n)} \cdot \sigma_z^{\otimes n} = 0 \ .
\end{equation}
To see this, we set up another, separate induction.
Clearly, the initial terms check out: $\{\partial^{(1)}, \ \sigma_z\} = \{\partial^{(2)}, \ \sigma_z^{\otimes 2}\} = 0$.
Now, the let the induction hypothesis be:
\begin{equation}
\{ \sigma_z^{\otimes n}, \ \partial^{(n)} \} = 0 \ ,
\end{equation}
and consider
\begin{align}
\nonumber
\left\{\sigma_z^{\otimes (n+1)}, \ \partial^{(n)} \right\} & = 
 \left\{\sigma_z^{\otimes (n+1)}, \ \partial^{(n)} \otimes I + \sigma_z^{\otimes n} \otimes Q^{+}  \right\}
 =
 \left\{\sigma_z^{\otimes n} \otimes \sigma_z, \ \partial^{(n)} \otimes I + \sigma_z^{\otimes n} \otimes Q^{+}  \right\}
 \\
 &=\left\{\sigma_z ^{\otimes n}, \ \partial^{(n)} \right\} \otimes (\sigma_z \cdot I) +
 (\sigma_z^{\otimes n} \cdot \sigma_z^{\otimes n}) \otimes \left\{ \sigma_z, \ Q^{+}  \right\} \ .
\end{align}
Now, the first term vanishes by induction hypothesis, and the second term also vanishes since $\left\{ \sigma_z, \ Q^{+}  \right\} = \{\partial^{(1)}, \ \sigma_z\} = 0$, the initial term. \end{proof}

\subsubsection{Hamiltonian}
We can also check the explicit form of the Hamiltonian  $\exp(i B^{(n)} t) = \exp\left(i (\partial^{(n)} + (\partial^{(n)})^\dagger 
)\right)$ against \eqref{expB}.

\paragraph{$n = 1$: }
Here, we have $B^{(1)} = \partial^{(1)} + (\partial^{(1)})^\dagger = 
\left( \begin{array}{cc} 0 & 1 \cr 1 & 0 \end{array} \right)$, so that
\begin{equation}
    \exp(i B^{(1)} t) = 
\left(
\begin{array}{cc}
 \cos (t) & -i \sin (t) \\
 -i \sin (t) & \cos (t) \\
\end{array}
\right) ~.
\end{equation}
On the right hand side of the theorem, we have $R = R_0 R_1$, where
$R_0 = I_2$ and
$R_1 = \exp(\frac12 Q_0 Q_1 \text{atan2(0,1)}) = I_2$.
Hence, the right hand side is just $\exp(-i X_0 t ) = \exp(-i \sigma_x t )$. Recalling that for integer $n$,
\begin{equation}
\sigma_x^{2n-1} = \sigma_x = \left( \begin{array}{cc} 0 & 1 \cr 1 & 0 \end{array} \right) ~, \qquad \sigma_x^{2n} = I_2 ~,
\end{equation}
we have equality.

\paragraph{$n = 2$: }
From the forms of $\partial^{(n)}$ we find:
\begin{equation}\label{left hand siden=2}
\exp(i B^{(2)} t) =
\left(
\begin{array}{cccc}
 \cos \left(\sqrt{2} t\right) & -\frac{i \sin \left(\sqrt{2} t\right)}{\sqrt{2}} & -\frac{i \sin \left(\sqrt{2} t\right)}{\sqrt{2}} & 0 \\
 -\frac{i \sin \left(\sqrt{2} t\right)}{\sqrt{2}} & \cos \left(\sqrt{2} t\right) & 0 & -\frac{i \sin \left(\sqrt{2} t\right)}{\sqrt{2}} \\
 -\frac{i \sin \left(\sqrt{2} t\right)}{\sqrt{2}} & 0 & \cos \left(\sqrt{2} t\right) & \frac{i \sin \left(\sqrt{2} t\right)}{\sqrt{2}} \\
 0 & -\frac{i \sin \left(\sqrt{2} t\right)}{\sqrt{2}} & \frac{i \sin \left(\sqrt{2} t\right)}{\sqrt{2}} & \cos \left(\sqrt{2} t\right) \\
\end{array}
\right)
 \ .
\end{equation}
We can now cross-check this explicit matrix against Theorem \ref{expB}.
We have, from $\partial^{(2)}$, that
$R = R_1 = \exp\left(\frac12 Q_0 Q_1 \text{atan2}(\sqrt{2-1},1) \right)$.
Thus,
\begin{align}
    R = R_1 = \exp\left(
    \frac{\pi}{8} (a_0 + a_0^\dagger)(a_1 + a_1^\dagger)
    \right) =
    \left(
\begin{array}{cccc}
 \cos \left(\frac{\pi }{8}\right) & 0 & 0 & \sin \left(\frac{\pi }{8}\right) \\
 0 & \cos \left(\frac{\pi }{8}\right) & \sin \left(\frac{\pi }{8}\right) & 0 \\
 0 & -\sin \left(\frac{\pi }{8}\right) & \cos \left(\frac{\pi }{8}\right) & 0 \\
 -\sin \left(\frac{\pi }{8}\right) & 0 & 0 & \cos \left(\frac{\pi }{8}\right) \\
\end{array}
\right) ~.
\end{align}
Hence,
\begin{equation}\label{right hand siden=2}
R^\dagger \exp(-i \sqrt{2} (a_0 + a_0^\dagger) t) R
=
\left(
\begin{array}{cccc}
 \cos \left(\sqrt{2} t\right) & -\frac{i \sin \left(\sqrt{2} t\right)}{\sqrt{2}} & -\frac{i \sin \left(\sqrt{2} t\right)}{\sqrt{2}} & 0 \\
 -\frac{i \sin \left(\sqrt{2} t\right)}{\sqrt{2}} & \cos \left(\sqrt{2} t\right) & 0 & -\frac{i \sin \left(\sqrt{2} t\right)}{\sqrt{2}} \\
 -\frac{i \sin \left(\sqrt{2} t\right)}{\sqrt{2}} & 0 & \cos \left(\sqrt{2} t\right) & \frac{i \sin \left(\sqrt{2} t\right)}{\sqrt{2}} \\
 0 & -\frac{i \sin \left(\sqrt{2} t\right)}{\sqrt{2}} & \frac{i \sin \left(\sqrt{2} t\right)}{\sqrt{2}} & \cos \left(\sqrt{2} t\right) \\
\end{array}
\right) ~.
\end{equation}
and we have perfect agreement of \eqref{left hand siden=2} and \eqref{right hand siden=2}.

\small

\end{document}